\documentclass[11pt]{article}

\usepackage{amsmath, amssymb, amsthm}
\usepackage{mathrsfs}
\usepackage{graphicx}
\usepackage{hyperref}
\usepackage[all]{xy}
\usepackage[a4paper,margin=2.5cm]{geometry}
\theoremstyle{plain}
\newtheorem{theorem}{Theorem}[section]
\newtheorem{lemma}[theorem]{Lemma}
\newtheorem{corollary}[theorem]{Corollary}
\newtheorem{proposition}[theorem]{Proposition}

\theoremstyle{definition}
\newtheorem{definition}[theorem]{Definition}
\newtheorem{example}[theorem]{Example}

\theoremstyle{remark}
\newtheorem{remark}[theorem]{Remark}

\numberwithin{equation}{section}

\title{Linear codes arising from geometrical operation}

\author{
Antonio Jesús Lorite López
\and
Daniel Camazón Portela
\and
Juan Antonio López Ramos\\
\normalsize Department of Mathematics, University of Almería\\
\normalsize Almería, Spain\\[6pt]
}

\date{} 

\begin{document}

\maketitle

\begin{abstract}
We construct linear codes over the finite field $\mathbb{F}_q$ from arbitrary simplicial complexes, establishing a connection between topological properties and fundamental coding parameters. First, we study the behaviour of the weights of codewords from a geometric point of view, interpreting them in terms of the combinatorial structure of the associated simplicial complex. This approach allows us to describe the minimum distance of the codes in terms of certain geometric features of the complex.

Subsequently, we analyse how various topological operations on simplicial complexes affect the classical parameters of the codes. This study leads to the formulation of geometric criteria that make it possible to explicitly control and manipulate these parameters.

Finally, as an application of the obtained results, we construct several families of optimal linear codes over $\mathbb{F}_2$ using these geometric methods. Thanks to the previously established geometric properties, we can precisely determine the parameters of these families.
\end{abstract}


\maketitle

\section{Introduction}

In recent decades, the interaction between coding theory and different areas of mathematics has given rise to new constructions and approaches of great interest. In particular, since the construction of linear codes from simplicial complexes was formalized \cite{Chang2018}, this topic has been the subject of study in numerous works (cf. \cite{Chang2018}, \cite{HuXuNi2024}, \cite{HyunJungyun2020}, \cite{LiuYu2023}, \cite{panliu2022}, \cite{SagarSarma2022LinearCodes}, \cite{SagarSarma2024}, \cite{WuZhuYue2019} for example). In these articles, the weights of the codewords have been analyzed mainly through formulas derived from the one introduced by Adamaszek \cite{adamszek}, based on the inclusion--exclusion principle, which allows one to obtain multivariable expressions for counting faces. This type of approach is largely disconnected from the geometry of the simplicial complex. The motivation of this work is to propose an alternative approach to the study of codewords and their weights, with the aim of providing a first geometric interpretation of weights in coding theory.

In \cite{HyunJungyun2020}, it is assumed that every maximal face of the simplicial complex must contain an isolated vertex; that is,
\begin{equation*}
A_i \setminus \bigcup_{j \in [s]\setminus\{i\}} A_j \neq \emptyset.
\end{equation*}
On the other hand, in \cite{HYUN2019135}, although the theoretical construction is generalized, in the study of concrete cases, attention is restricted to complexes with a unique maximal face in order to control the weight function and analyze its properties. In general, in previous works, the imposed hypotheses prevent addressing the study of codes associated with simplicial complexes with a richer geometric or topological structure, such as, for example, triangulations of (pseudo)varieties. The approach of this article aims to extend the knowledge of the code parameters to the case of arbitrary simplicial complexes.

To this end, we begin by presenting a series of preliminary notions that establish how the code will be defined to reflect the topological structure of the complex. Likewise, the notation that will be used throughout the article is introduced. Subsequently, in Section 3, one of the key results of this work is presented: Theorem \ref{Theorem 3.1}, which makes it possible to understand the weight in terms of intersections of the faces of the complex. From this result and the geometric approach to the weights of codewords, the remaining results are developed. One of the main advantages of this approach is that it allows one to dispense with the classical weight function, which becomes too complex in the case of arbitrary simplicial complexes. Following this perspective, a series of topological operations on complexes and how they affect the code are studied in Section 4. Finally, working over the field $\mathbb F_2$ to further facilitate geometric intuition, the construction of some families of codes is presented based on the results obtained previously.

From a broader point of view, this work can be interpreted as a first step toward a possible interaction between coding theory, algebraic geometry and ring theory, mediated by topological and geometric tools. The results obtained suggest that the structure of simplicial complexes can serve as a natural setting to explore new connections between both disciplines. In this sense, the present article aims to contribute to the construction of a common language that allows the transfer of intuitions and techniques between topology, geometry, and coding theory.

\section{Preliminaries}
The aim of this section is to introduce the main notions and notation that will be used throughout the paper.
\begin{definition}
A simplicial complex $\Delta$ is a family of subsets of $[n]=\{1,\dots,n\}$, such that if $\sigma \in \Delta$ and $\tau \subseteq \sigma $, then $\tau \in \Delta$.
\end{definition}

A linear code over $\mathbb{F}_q$ can be constructed from a simplicial complex $\Delta$. Let $\sigma \in \Delta$. We can identify $\sigma$ with its characteristic vector in $\mathbb{F}_q$, that is, $\sigma \sim \chi_\sigma \in \{0,1\}^n$, where

\begin{align*}
    \chi_\sigma(i) = \begin{cases}
            1, \text{ if } i \in \sigma,\\
            0, \text{ if } i \not \in \sigma,
           \end{cases} 
           \forall i \in [n]
\end{align*}

We then define $D := \{ \chi_\sigma \in \mathbb{F}_q^n : \sigma \in \Delta \}\subseteq \mathbb F_q^n$. The linear code associated with $\Delta$ can be defined as follows
\begin{equation*}
    C_\Delta=\lbrace c_\Delta(u)= (u\cdot x)_{x\in D} : u \in \mathbb{F}_q^n \rbrace,
\end{equation*} 
where $u \cdot x$ denotes the coordinatewise or Hadamard product of $u$ and $x$.

Equivalently, we can describe the definition matrix $G$ of the code as one whose rows are the elements of $D$; in that case:
\begin{equation*}
    C_\Delta=\lbrace uG^T : u \in \mathbb{F}_q^n \rbrace
\end{equation*}
Note that $G^T$ is a generator matrix of the code $C_\Delta$.

\begin{definition}
    A maximal simplex $A$ of $\Delta$ is a simplex that is not properly contained within any other simplex of the complex.
\end{definition} 

Thus, if $A_1,\dots,A_s$ are the maximal simplices of $\Delta$, we may write
\[
\Delta = \langle A_1,\dots,A_s \rangle
= \{ B \subseteq \{1,\dots,N\} \mid B \subseteq A_i \text{ for some } i \in \{1,\dots,s\} \}.
\]

\begin{definition}
    The complementary simplicial complex $ \Delta^c$ is defined by 
\begin{equation*}
    \Delta^c= \{ \sigma \in [n]: \sigma \not \in \Delta\}.
\end{equation*}
\end{definition}

The linear code $C_{\Delta^c}$ is called the anticode of $\Delta$.

\begin{definition}
    Let $\Delta$ be a simplicial complex. The link of a vertex, $lk_\Delta(e_i)$, is a collection of subsets satisfying $lk_\Delta(e_i)=\{ \tau \subset \Delta: \tau \cup  \{ e_i\}\in \Delta \}$. 
\end{definition}

\section{Geometrical and topological meaning of code parameters}
Let $\Delta$ be a simplicial complex and let $C_\Delta$ be its associated code,
\begin{equation*}
      C_\Delta=\lbrace (u\cdot x)_{x\in D} : u \in \mathbb{F}_q^n \rbrace.
\end{equation*}
From a geometric point of view, the product $(u \cdot x)$ can be understood as
\begin{equation*}
  (u\cdot x)=\sum_{v \in \sigma }u_v,
\end{equation*}
where $x$ is the characteristic vector of $\sigma$. Following this approach, it is possible to define
\begin{equation*}
l_\sigma(u) = \sum_{v \in \sigma} u_v,    
\end{equation*}
so that
\begin{equation*}
  c_\Delta(u)=(l_\sigma(u))_{\emptyset \neq \sigma \in \Delta}
\end{equation*}
and 
\begin{equation*}
  w(c_\Delta(u))=\#\{ \sigma \in \Delta: l_\sigma(u)) \not \equiv 0 \pmod{q} \}.
\end{equation*}
That is, the weight of a codeword can be interpreted as the number of simplices of the complex for which the sum over their vertices does not vanish on $u$.
Based on this perspective, the following result is presented:

\begin{theorem}\label{Theorem 3.1}
    Let $\Delta$ be a simplicial complex on a vertex set $V=\{1,\dots,k\}$ and let
    $C_\Delta$ be the associated linear code over $\mathbb{F}_q$.
    For any $u\in\mathbb{F}_q^k$ with $|\operatorname{supp}(u)|\ge 2$, there exists
    $u'\subset u $ such that $|supp(u')|=1$ and
    \[
    w(c_\Delta(u')) \le w(c_\Delta(u)).
    \]
\end{theorem}

\begin{proof} 
    Fix $u \in \mathbb{F}_q^k : u=\{ u_1,...,u_k\}$ and $u_i\in \mathbb{F}_q^*$. 
    For $\sigma\in\Delta$, we can extend
    \[
    \tilde l_\sigma(u)=\sum_{v\in\sigma}\tilde u_v\in\mathbb Z,
    \]
    where $\tilde{u_v}$ is the lifting of $u_v$ into $\mathbb{Z}$.
    
     Recall $ w(c_\Delta(u))=\#\{ \sigma \in \Delta: l_\sigma(u)) \not \equiv 0 \pmod{q} \}$ and $c_\Delta(u)=(l_\sigma(u))_{\emptyset \neq \sigma \in \Delta}\in \mathbb{F}_q^n$, where $n=|\Delta|-1$. We define
    \[
    \Delta^a=\{ \tau \in \Delta : \sum_{v\in\tau}\tilde{u_v}=a\}.
    \]
    and
    \[
    \Delta^a_{v_i}=\{ \tau \in \Delta : \sum_{v\in\tau}\tilde{u_v}=a \hspace{2mm}\land \hspace{2mm} v_i \in \tau\}.
    \]
    Therefore, $w(c_\Delta(u))=\sum_{a\neq \mathbb{Z}q}|\Delta^a|$.
    Let $u'\subset u $ such that $supp(u')=v_i \in \mathbb{F}_q^*$. Since $u'$ is supported in a single vertex $v_i$, for any face $\tau$ containing $v_i$, we have $\tilde l _\tau(u')=\tilde u_i \not \in \mathbb{Z}q$, hence all such faces contribute to the weight.
    \[
    w(c_\Delta(u'))=\sum|\Delta^a_{v_i}|
    .\]
    We use the decomposition
    \[
    |\Delta^a| = |\Delta^a_{v_i}| + |\Delta^a_{\neg v_i}|
    \qquad \text{for all } a\in\mathbb Z.
    \]
    
    Therefore,
    \[
    \sum_{a\notin q\mathbb Z} |\Delta^a|
    =
    \sum_{a\notin q\mathbb Z} |\Delta^a_{v_i}|
    +
    \sum_{a\notin q\mathbb Z} |\Delta^a_{\neg v_i}|.
    \]
    
    Since
    \[
    w(c_\Delta(u'))=\sum_{a\in\mathbb Z}|\Delta^a_{v_i}|,
    \qquad
    w(c_\Delta(u))=\sum_{a\notin q\mathbb Z}|\Delta^a|,
    \]
    the inequality \(w(c_\Delta(u'))\le w(c_\Delta(u))\) is equivalent to
    \[
    \sum_{a\in q\mathbb Z}|\Delta^a_{v_i}|
    \le
    \sum_{a\notin q\mathbb Z}|\Delta^a_{\neg v_i}.
    \]
    
    Now, for each $t=pq$, define a map
    \begin{gather*}
    \Phi_t:\Delta^{pq}_{v_i}\to\Delta^{pq-\tilde u_i}_{\neg v_i},
    \qquad
    \tau\mapsto\tau\setminus\{v_i\}.
    \end{gather*}
    This map is well defined as $\tau \in \Delta$ implies $\tau \setminus \{ v_i\} \in \Delta$. Moreover, $pq-\tilde u_i\notin\mathbb Z q$ since $\tilde u_i\notin\mathbb Z q$, and $p_1q-\tilde u_i=p_2q-\tilde u_i$ if and only if $p_1=p_2$.The map is also injective.  Therefore, $|\Delta^{(pq)}_{v_i}|\leq |\Delta^{(pq-\tilde u_i)}_{\neg v_i}|$.
    
     Summing over all $t$, we obtain
    \[
    \sum_{j=\mathbb{Z}q}|\Delta^{j}_{v_i}|
    \le
    \sum_{j = \mathbb{Z}q}|\Delta^{j-\tilde u_i}_{\neg v_i}|\le
    \sum_{j \neq \mathbb{Z}q}|\Delta^j_{\neg v_i}|.
    \]
    
    Hence, $w(c_\Delta(u)) \ge w(c_\Delta(u'))$. 
\end{proof}

\begin{corollary}\label{Corollary 3.2.}
     The minimum weight of the code, and hence its minimum distance, is determined by vectors $u \in \mathbb{F}_q^n$ whose support consists of a single vertex of the complex.
\end{corollary}
Another important consequence that can be drawn from this result is that the distance of the code does not depend on the chosen value of $q$. Assuming that $u$ has unit support, one can see that, since
\[
l_\sigma(u) \not\equiv 0 \pmod{q} \iff {supp}(u) \cap \sigma \neq \emptyset,
\]
it follows that
\begin{equation*}
w(c_\Delta(u)) = \#\left\{ \sigma \in \Delta : supp(u) \cap \sigma \neq \emptyset \right\}.
\end{equation*}

Therefore, the value of the weight depends only on the support of $u$.

\begin{corollary}
For a simplicial complex $\Delta$ of $\mathbb{F}_q^n$
whose set of maximal elements is $\{A_1,\dots,A_s\}$ and such that
$\cup_{i=1}^s A_i = [k]$,
if each $A_i$ satisfies
$A_i \setminus \bigcup_{j \in [s]\setminus i}A_j\neq \emptyset$,
then the parameters of $C_{\Delta^*}$ are

\begin{equation*}
    [|\Delta|-1, k, \min\{ 2^{|A_i|-1} : i \in [s]\}].
\end{equation*}

In general, for an arbitrary simplicial complex $\Delta$, the parameters of $C_{\Delta^*}$ are

\begin{equation*}
    [|\Delta|-1, k, \min\{ lk(i) : i \in [k]\}].
\end{equation*}
\end{corollary} 

\section{Construction of new codes via topological operations}

\begin{proposition} 
    Let $\Delta = \Delta_1 \sqcup \Delta_2$ be a simplicial complex over $\mathbb F_q$
    with two disjoint connected components, and let $\widetilde{\Delta}$ be the simplicial
    complex obtained by identifying a face $F_1 \subset \Delta_1$ with a face
    $F_2 \subset \Delta_2$ of the same dimension.
    Then the minimum distance of the associated linear code does not decrease under
    the identification, that is,
    \[
    d(\widetilde{\Delta}) \ge d(\Delta).
    \]
\end{proposition}
\begin{proof}
    We regard $\widetilde{\Delta}$ as the quotient of $\Delta$ obtained by identifying
    the faces $F_1$ and $F_2$, and for each simplex $\sigma \in \Delta$ we denote by
    $\widetilde{\sigma}$ its image in the quotient.
    
    By Corollary \ref{Corollary 3.2.}, the minimum-weight codewords of the code associated to a simplicial
    complex are precisely those whose support consists of a single vertex.
    There are two possible cases.
    
    First, suppose that there exists a minimum-weight codeword whose support is a vertex
    that is not identified in the quotient. In this case, the identification does not
    affect the set of simplices intersecting the support; therefore, the weight of
    the corresponding codeword remains unchanged.
    
    Otherwise, every vertex supporting a minimum-weight codeword belongs to one of the
    faces that are identified.
    Since such codewords have support consisting of a single vertex, we have
    \begin{equation*}
    w(c_{\widetilde{\Delta}}(\widetilde{u}))
    =
    \#\left\{ \widetilde{\sigma} \in \widetilde{\Delta}
    :\operatorname{supp}(\widetilde{u}) \cap \widetilde{\sigma} \neq \emptyset \right\}
    \ge
    \#\{ \sigma \in \Delta
    :\operatorname{supp}(u) \cap \sigma \neq \emptyset \}=w(c_{{\Delta}}({u})).
    \end{equation*}
    Indeed, since $\Delta_1$ and $\Delta_2$ are disjoint connected components, no simplex
    of $\Delta$ contains vertices from both components.
    Therefore, under the identification, intersections with simplices in the original
    component are preserved, while additional intersections may only arise from simplices
    in the other component.
    This shows that the weight cannot decrease, and hence
    $d(\widetilde{\Delta}) \ge d(\Delta)$.
\end{proof}

\begin{remark}
    In the situation of the previous proposition, suppose that every minimum--weight
    codeword of $C_\Delta$ is supported on a vertex belonging to one of the faces that are
    identified.
    Let $u$ be such a codeword and let $\widetilde{u}$ denote its image in the quotient.
    
    In this case, the increase of the weight produced by the identification can be described
    explicitly.
    Writing $\widetilde{\Delta}=\widetilde{\Delta}_1\cup\widetilde{\Delta}_2$, we have
    \[
    \mathrm{wt}(\widetilde{u}) = \mathrm{wt}(\widetilde{u}_{|\widetilde{\Delta}_1}) + \mathrm{wt}(\widetilde{u}_{|\widetilde{\Delta}_2}) - \mathrm{wt}(\widetilde{u}_{|F_1}),
    \]
    and since $\mathrm{wt}(\widetilde{u}_{|\widetilde{\Delta}_1})=d(\Delta)$, the contribution
    to the increase comes exclusively from simplices in $\Delta_2$ containing the
    identified face.
    More precisely,
    \[ \mathrm{wt}(\widetilde{u}) = d(\Delta) + \sum_{S \subset \Delta_2 \,:\, F_2 \subset S} \left( \mathrm{wt}(\widetilde{u}_{|S}) - \mathrm{wt}(\widetilde{u}_{|F_2}) \right).
    \]
\end{remark}

\begin{example}
    Let $\Delta$ be a simplicial complex over $\mathbb F_2$ whose set of maximal faces is \\ $\{\{0,1\},\{0,2,3\},\{4,5,6,7\}\}$. To determine the length of the code, it suffices to count the number of non-empty simplices of $\Delta$, since the generator matrix has one column for each such simplex. Therefore, multiplying a vector of dimension $k$ by this matrix yields a codeword whose length coincides with the number of columns. The dimension of the code is equal to the number of vertices of $\Delta$, which is equal to the number of rows in the generating matrix. To compute the minimum distance, we rely on Corollary $2.2$. Vertex $\{1\}$ is the one contained in the smallest number of simplices, as it only appears in the simplex $\{1\}$ and in the edge $\{0,1\}$. Hence, the minimum distance of the code is $2$. Collecting these results, the code $C_\Delta$ has parameters $[24,8,2]$. Meanwhile, if we identify the vertices according to the relations $1 \sim 4$, $2 \sim 5$, and $3 \sim 6$, we obtain a simplicial complex $\tilde{\Delta}$. By an argument analogous to the one above, we conclude that the associated code $C_{\tilde{\Delta}}$ has parameters $[20,5,5]$.
\end{example}
\begin{figure}[h]
  \centering
  \includegraphics[width=0.6\textwidth]{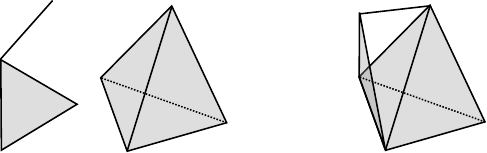}  
  \caption{Example of vertex gluing}
  \label{fig:identification}
\end{figure}

\begin{proposition} \label{Proposition 4.3}
    Let $\Delta$ be a simplicial complex and let $cone(\Delta)$ denote the cone over $\Delta$ with apex $c$.
    Let $u$ be a nonzero vector of the code associated to $cone(\Delta)$. If $c \notin supp(u)$, then $w_{cone(\Delta)}(u)=2\,w_{\Delta}(u_{|\Delta})$. If $c \in supp(u)$, then $w_{cone(\Delta)}(u)\geq |\Delta|-1$.
\end{proposition}

\begin{proof} 
    By definition, the cone over $\Delta$ is given by $
    cone(\Delta) =  \{ \langle \mathcal{F}_i \cup \{c\}\rangle \}
    $, where $\mathcal{F}_i$ are the maximal elements of $\Delta$.
    
    Assume first that $c \notin \operatorname{supp}(u)$.
    For any simplex $\sigma \in \Delta$, we have
    \[
    \sum_{v\in \sigma}u_v\equiv \sum_{v\in \sigma\cup \{ c \}}u_v\pmod q. 
    \]
    Therefore, $\sigma$ contributes to the weight of $u$ if and only if
    $\sigma \cup \{c\}$ does.
    It follows that each contributing simplex in $\Delta$ gives rise to exactly two
    contributing simplices in $cone(\Delta)$, and hence
    \[
    w_{cone(\Delta)}(u) = 2\,w_{\Delta}(u_{|\Delta}).
    \]
    Assume now that $c \in supp(u)$.
    For every simplex $\sigma \in \Delta$, at least one of the two simplices
    $\sigma$ or $\sigma \cup \{c\}$ satisfies
    \[
    |\tau \cap supp(u)| \not\equiv 0 \pmod q.
    \]
    Consequently, at least one simplex in each pair
    $\{\sigma,\sigma \cup \{c\}\}$ contributes to the weight.
    Since $\Delta$ has $|\Delta|$ simplices and the empty simplex does not contribute,
    we obtain
    \[
    w_{\operatorname{cone}(\Delta)}(u)\geq |\Delta|-1.
    \] 
\end{proof}

\begin{corollary}
     Let the code associated to $\Delta$ have parameters $[n,k,d]$. Then the code associated with $cone(\Delta)$ has a minimum distance of $d'=2d$. Furthermore,  the length becomes $n'=2n+1$ and the dimension $k'=k+1$.    
\end{corollary}

\begin{example}
    Let $\Delta$ be a simplicial complex whose set of maximal faces is $\{ \{0,1,2\},\{0,1,3\},$ $\{0,2,3\},\{1,2,3\}\}$, the 2-skeleton of a 3-simplex. The associated code $C_{\Delta^*}$ has parameters $[14,4,7]$. If we take the cone over $\Delta$, we obtain a new code with parameters $[29,4,14]$. 
\end{example}

If all maximal faces contain a common vertex $c$, it is possible to define an operator that acts as an inverse of the operation $cone$ as the composition of elementary collapses with respect to that vertex. More precisely, given a simplicial complex containing a maximal face $A_i=\langle v_1,\dots,v_n,c\rangle$, that maximal face and the free faces containing the common vertex $c$ are removed,
\begin{equation*}
  \langle v_1,\dots,v_n,c\rangle \longrightarrow \bigcup_{i\in[n]} \langle v_1,\dots,\widehat{v_i},\dots,v_n\rangle,
\end{equation*}
where the symbol $\widehat{v_i}$ denotes the omission of the vertex $v_i$. Repeating this process for all maximal faces of the complex yields a simplicial complex that can be interpreted as the inverse of the operation $cone$.

In the general case, when there is no vertex common to all maximal faces, one can define an analogous operator that reduces the dimension of the simplicial complex, namely, the boundary operator. Given an arbitrary maximal face, one can apply the inverse operation of the operator $cone$ by removing the free faces containing each of its vertices. By performing this process for all vertices of the maximal face, that is, by means of a combination of elementary collapses, the boundary of the maximal face is obtained. Repeating this procedure for all maximal faces of the complex yields the boundary of the simplicial complex.

\begin{proposition}
    Let $\Delta$ be a simplicial complex, and let $\partial(\Delta)$ denote its boundary complex, that is, the subcomplex consisting of all non-maximal simplices of $\Delta$. Passing from $\Delta$ to $\partial(\Delta)$ increases the minimum distance of the associated anticode, whereas for the corresponding linear code the minimum distance does not increase.

    Moreover, if $\Delta$ has $s$ maximal faces and no isolated vertices, then the associated codes satisfy
    \[
    d(C_{\Delta^*}) \ge d(C_{\partial(\Delta)^*}),
    \]
    the length decreases from $n$ to $n' = n - s$, and the dimension remains unchanged.
\end{proposition}

\begin{proof}
    Let $u'$ be a vector supported on $\partial(\Delta)$. Its weight is given by
    \[
    w(u') = \#\left\{ \partial(\sigma) : \sigma \in \Delta \text{ and } \sum_{v \in \partial(\sigma)} u_v \not\equiv 0 \pmod q \right\},
    \]
    or equivalently,
    \[
    w(u') = \#\left\{ \sigma \in \Delta : \text{$\sigma$ is not maximal and } \sum_{v \in \sigma} u_v \not\equiv 0 \pmod q \right\}.
    \]
    
    When passing from $\Delta$ to $\partial(\Delta)$, the only contributions that may be lost correspond to maximal simplices. Since there are $s$ such simplices, the loss in weight is bounded by
    \[
    w(u) - w(u') \le s,
    \]
    which implies
    \[
    d(C_{\Delta^*}) \ge d(C_{\partial(\Delta)^*}).
    \]
    
    Finally, removing the $s$ maximal simplices reduces the length of the code from $n$ to $n - s$. The dimension remains unchanged provided that $\Delta$ has no isolated vertices.
    
    When passing from $\Delta$ to its boundary $\partial(\Delta)$, the maximal faces of $\Delta$ become elements of the complement $\Delta^c$. Consequently, the generator matrix of the associated anticode acquires $s$ new columns corresponding to these maximal faces. Since none of the original columns are removed, the weight of the anticode can only increase.
\end{proof}

\begin{lemma}
     Let $k>1$, and let $u\in\mathbb F_q^k$ be a nonzero vector.
    Then
    \[
    \#\{x\in\mathbb F_q^k:\langle u,x\rangle=0\}=q^{k-1},
    \]
    and consequently
    \[
    \#\{x\in\mathbb F_q^k:\langle u,x\rangle\neq 0\}=(q-1)q^{k-1}.
    \]    
\end{lemma}

\begin{proof}
    Since $u\neq 0$, the map
    \[
    \varphi_u:\mathbb F_q^k\longrightarrow\mathbb F_q,
    \qquad
    x\longmapsto\langle u,x\rangle
    \]
    is a nonzero linear functional. Its kernel is therefore a subspace of codimension
    $1$, and hence has cardinality $q^{k-1}$. The result follows.
\end{proof}

\begin{theorem}
    Let $\{\Delta_k\}_{k> 1}$ be a family of simplicial complexes on vertex sets
    $[k]$, and assume that there exists $n_0\in\mathbb N$ such that
    \[
    \dim\Delta_k\le n_0
    \quad\text{for all }k.
    \]
    Let $ \Delta^c_k=\mathbb F_q^{k}\setminus\{\mathbf\chi_\sigma:\sigma\in\Delta_k\} $ denote the family of complementary simplicial complexes. Then, for every nonzero $u\in\mathbb F_q^{k}$,
    \[
    w(u) = \#\{x\in\Delta^c_k:\langle u,x\rangle\neq 0\} = (q-1)q^{k-1}+\mathcal O(k^{n_0+1}),
    \]
    where the implicit constant is independent of $u$. In particular,
    \[
    \frac{d(\Delta^c_k)}{|\Delta^c_k|}
    \longrightarrow
    \frac{q-1}{q}
    \qquad
    \text{as }k\to\infty.
    \]
    The relative minimum distance converges to $(q-1)/q$, which is the maximal possible value for the relative distance of a $q$-ary code.
\end{theorem}

\begin{proof} 
    Fix $k$ and let $u\in\mathbb F_q^{k}$ be nonzero. By the previous lemma,
    \[
    \#\{x\in\mathbb F_q^{k}:\langle u,x\rangle\neq 0\}
    =
    (q-1)q^{k-1}.
    \]
    
    The anticode $\Delta^c_k$ is obtained from $\mathbb F_q^{k}$ by removing
    the set $E_k=\{\mathbf \chi_\sigma:\sigma\in\Delta_k\}$. Since $\dim\Delta_k\le n_0$, the number of simplices satisfies $|E_k|=|\Delta_k|=\mathcal O(k^{n_0+1})$.
    
    Removing the set $E_k$ can affect the weight by at most $|E_k|$, and hence
    \[ w(u) = (q-1)q^{k-1}-\#\{ x\in E_k: \langle x,u\rangle \neq 0\}= (q-1)q^{k-1} + \mathcal O(k^{n_0+1}).
    \]
    
    Since $|\Delta^c_k|=q^k-\mathcal O(k^{n_0})$, dividing by
    $|\Delta^c_k|$ yields
    \[
    \frac{d( \Delta^c_k)}{|\Delta^c_k|} \xrightarrow[k\to\infty]{} \frac{q-1}{q}.
    \]
\end{proof}

\begin{example}
    Let  $\Delta_0 = \langle \{1,2,3\},\{3,4,5\},\{1,3,5\}\rangle$ be a simplicial complex on the vertex set $\{1,2,3,4,5\}$. By performing suitable subdivisions of the simplices of $\Delta_0$, we obtain a refined triangulation
    \[
    \Delta = \langle 
    \{1,2,8\},\{1,3,8\},\{2,3,8\},
    \{3,4,6\},\{3,5,6\},\{4,5,6\},
    \{1,5,7\},\{1,3,7\},\{3,5,7\}
    \rangle
    \]
    on the vertex set $\{1,2,3,4,5,6,7,8\}$. 
    
    Consider the complementary set
    $\Delta^c \subset \mathbb{F}_3^{8}$, obtained as the complement of the set of
    characteristic vectors of the simplices of $\Delta$, and the corresponding
    linear code $C(\Delta^c)$.
    
    A direct computation shows that this code has parameters
    \[
    [6527,\,8,\,4344].
    \]
    In particular, the relative minimum distance satisfies
    \[
    \frac{4344}{6527} \approx 0.6655,
    \]
    which is very close to the asymptotic value $(q-1)/q = 2/3$ predicted by the previous theorem.
\end{example}

\begin{example}
      Continuing with the sets from the previous example, the associated anticode to $\Delta_0$ over $\mathbb{F}_2$ has parameters 
    \[
    [16,5,6],
    \]
    giving a distance-to-length ratio of $6/16 = 0.375$.
    
    Whereas the anticode over $\mathbb{F}_2$ associated with $\Delta$ has parameters 
    \[
    [222,8,107],
    \]
    with a distance-to-length ratio of approximately $107/222 \approx 0.482$, which is a significant improvement over the original complex.
\end{example}

\section{Construction of codes over $\mathbb F_2$ using geometric properties}

In this section, we focus on the construction of linear codes over the field
$\mathbb{F}_2$ using the geometric properties of simplicial complexes developed
in the previous sections. There are several reasons for choosing $\mathbb{F}_2$
as the base field.

First, as observed earlier, for this class of codes, the minimum distance is
independent of the size of the field. Moreover, working over $\mathbb{F}_2$
allows us to interpret the weight of a codeword from a geometric point of view, i.e. in
terms of intersections with simplices rather than sums modulo $q$. Finally,
from a computational point of view, constructions over $\mathbb{F}_2$ are
significantly less expensive and, therefore, more amenable to explicit
calculations and experimentation.

Let $\Delta$ be a simplicial complex on the vertex set $[N]$. For any vector
$u \in \mathbb{F}_2^k$, the weight of the corresponding codeword admits a
geometric interpretation as the number of simplices of $\Delta$ whose
intersection with the support of $u$ has odd cardinality. That is,
\begin{equation*}
    w(c_{\Delta^*}(u))=\# \left\{ \sigma \in \Delta: |supp(u) \cap \sigma| \equiv 1 \pmod 2 \right\}
\end{equation*}

We begin with the simplest possible example. Let $\Delta$ be a simplicial
complex consisting of a single maximal face $A_1$ with
$|A_1| = N+1$. Equivalently, $\Delta$ is the standard $N$--simplex.

From the results established above in Corollary \ref{Corollary 3.2}, it follows directly that the associated
code $C_{\Delta^*}$ has parameters
\[
[2^{N+1}-1,\, N,\, 2^{N}].
\]
In particular, all codes arising from simplicial complexes of this form are
optimal.

\begin{proposition}
    The linear code associated with the $(N-1)$--skeleton of an $N$--simplex has
    parameters
    \[
    [2^{N+1}-2,\, N+1,\, 2^{N}-1].
    \]
    In particular, these codes are length--optimal. 
\end{proposition}

\begin{proof}
    As shown in the previous section, passing to the boundary of a simplicial
    complex decreases the length of the associated code by the number of maximal
    faces of the original complex, and decreases the minimum distance by at most
    one per maximal face removed. Since an $N$--simplex has exactly one maximal
    face, both the length and the minimum distance decrease by exactly one.
    
    Therefore, the resulting code has parameters
    \[
    [2^{N+1}-2,\, N+1,\, 2^{N}-1].
    \]
    
    It remains to verify that these parameters attain the Griesmer bound.
    \begin{align*}
    \sum_{n=0}^{N} \left\lceil \frac{2^{N}-1}{2^{n}} \right\rceil
    &= 2^{N}-1 + \sum_{n=1}^{N} \left\lceil \frac{2^{N}-1}{2^{n}} \right\rceil \\
    &= 2^{N}-1 + \sum_{n=1}^{N} \left\lceil \frac{2^{N}}{2^{n}}- \frac{1}{2^{n}}  \right\rceil  \\
    &= 2^{N}-1 + \sum_{n=1}^{N} \frac{2^{N}}{2^{n}}  = 2^{N}-1 +2^{N}-1=2^{N+1}-2.
    \end{align*}
    
    This shows that the code meets the Griesmer bound and is therefore optimal.
\end{proof}

\begin{theorem}
    An infinite family of optimal linear codes can be constructed by taking the cone over the $(N\!-\!1)$--skeleton of an $N$--simplex. These codes have parameters
    \[
    [\,2(2^{N+1}-2)+1,\; N+2,\; 2(2^{N}-1)].
    \]
\end{theorem}

\begin{proof}
    From the previous result, the linear code associated with the $(N\!-\!1)$--skeleton of an $N$--simplex has parameters
    \[
    [\,2^{N+1}-2,\; N+1,\; 2^{N}-1\,].
    \]
    Taking the cone over this simplicial complex modifies the parameters in a controlled way, according to the geometric results established in Proposition \ref{Proposition 4.3}. In particular, the length is doubled and increased by one, the dimension increases by one, and the distance is doubled, yielding a new code with parameters
    \[
    [\,2(2^{N+1}-2)+1,\; N+2,\; 2(2^{N}-1)\,].
    \]
    
    We now analyze the Griesmer bound. A direct computation shows that
    \begin{align*}
    \sum_{n=0}^{N+1} \left\lceil \frac{2(2^{N}-1)}{2^{n}} \right\rceil
    &= 2(2^{N}-1) + (2^{N}-1)
       + \sum_{n=1}^{N} \left\lceil \frac{2^{N}-1}{2^{n}} \right\rceil \\
    &= 2(2^{N}-1) + (2^{N}-1)
       + \sum_{n=1}^{N} \frac{2^{N}}{2^{n}} \\
    &= 2(2^{N}-1) + (2^{N}-1) + (2^{N}-1) \\
    &= 2(2^{N+1}-2).
    \end{align*}
    Therefore, the length of the code exceeds the Griesmer bound by exactly one unit.
    
    However, no code with the same length and dimension can have strictly larger minimum distance. Indeed, assume that there exists a code with parameters $[\,2(2^{N+1}-2)+1,\; N+2,\; 2(2^{N}-1)+1\,]$.
    Then
    \begin{align*}
    \sum_{n=0}^{N+1} \left\lceil \frac{2(2^{N}-1)+1}{2^{n}} \right\rceil
    &= 2(2^{N}-1)+1 + (2^{N}-1)+1 
       + \sum_{n=1}^{N} \left\lceil \frac{2^{N}-1}{2^{n}} \right\rceil \geq  \\
    &2 + \sum_{n=0}^{N+1} \left\lceil \frac{2(2^{N}-1)}{2^{n}} \right\rceil >\\
    & 2(2^{N+1}-2)+1,
    \end{align*}
    which contradicts the Griesmer bound. Hence, no code with these parameters can have larger minimum distance, and the constructed code is distance–optimal.
\end{proof}

\section*{Acknowledgments}
\noindent This research is supported by PID 2022-140934OB-I00 funded by MI-CIU/AEI/10.13039/501100011033 and ``Junta de Andalucía FQM-425''. The second author was partially supported by grant PID2022-138906NB-C21 funded by MI-CIU/AEI/10.13039/501100011033 and by ERDF A way of making Europe.

\bibliographystyle{plain}   
\nocite{*}
\bibliography{refences}

\end{document}